\documentclass[11pt]{article}

\usepackage[T1]{fontenc}
\usepackage[utf8]{inputenc}
\usepackage{lmodern}
\usepackage{amsmath,amssymb,amsthm,mathtools}
\usepackage{booktabs}
\usepackage{enumitem}
\usepackage{hyperref}
\usepackage{cleveref}
\usepackage[a4paper,margin=1in]{geometry}

\hypersetup{
  colorlinks=true,
  linkcolor=blue,
  citecolor=blue,
  urlcolor=blue
}

\newtheorem{theorem}{Theorem}[section]
\newtheorem{definition}[theorem]{Definition}
\newtheorem{proposition}[theorem]{Proposition}
\newtheorem{corollary}[theorem]{Corollary}
\newtheorem{lemma}[theorem]{Lemma}
\newtheorem{remark}[theorem]{Remark}

\title{FDY-privacy: Graded Anonymity and Unlinkability}
\author{Murat Moran}
\date{\today}

\newcommand{\DedF}{\operatorname{Ded}_{\mathrm{FuzzyDY}}}
\newcommand{\DedC}{\operatorname{Ded}_{\mathrm{DY}}}
\newcommand{\ModelDY}{\mathcal{M}_{DY}}
\newcommand{\ModelFuzzy}{\mathcal{M}_{FuzzyDY}}

\title{Graded Symbolic Verification with a Fuzzy Dolev-Yao Attacker Model}
\author{Murat Moran}
\date{\today}

\begin{document}
\maketitle

\begin{abstract}
Classical symbolic protocol verification under Dolev--Yao uses binary attacker knowledge (known/unknown). This abstraction misses cumulative side-channel settings, where repeated noisy observations progressively improve attacker knowledge. We model this process with a graded attacker view \(\mu_K\in[0,1]\), product T-norm leak updates, and finite-grid explicit-state execution in Modified Murphi.

The method is optimised with exact concept-lattice attribute reducts and exposes threshold-driven safe-to-fail transitions that are not represented in corresponding binary runs under the same bounded assumptions. Executed results on symmetric and asymmetric protocols, including Needham--Schroeder--Lowe (NSL), show that baseline models passing under crisp semantics can fail once cumulative side-channel leakage is enabled.
\end{abstract}


\maketitle

\section{Introduction}
\label{sec:intro}
The Dolev--Yao (DY) model~\cite{dolev1983security} is a standard basis for symbolic protocol verification under perfect cryptography. Its attacker knowledge semantics are binary (known/unknown). However, in side-channel settings, observations are partial, noisy, and cumulative. This mismatch can hide attack paths that emerge only after multiple weak observations. A canonical historical example is Kocher's timing attack line, which showed that repeated timing measurements can leak private-key information in RSA implementations~\cite{kocher1996timing}.

The method in this paper extends symbolic analysis with a graded attacker view. High-assurance tools such as FDR-based checkers~\cite{fdr}, ProVerif~\cite{blanchet2016proverif}, Maude-NPA~\cite{Escobar2009}, and Murphi~\cite{dill1996murphi} are effective for crisp derivability, but they do not natively model incremental knowledge-quality updates inside the attacker state. The Fuzzy Dolev-Yao (FuzzyDY) extension addresses this by replacing binary possession with a knowledge map \(\mu:\mathcal{M}\rightarrow[0,1]\), lifting deduction via Zadeh's Extension Principle~\cite{zadeh1965fuzzy}, and applying T-norm-based leak updates.

This graded view is needed for three attack classes that crisp DY cannot represent directly. \textbf{Cumulative sub-threshold observations} capture multi-step evidence accumulation, where no single observation is sufficient, but the sequence is. \textbf{Near-threshold attack transitions} capture safe-to-fail behaviour after a small additional observation update. \textbf{Ranked attack feasibility} enables comparing attack paths by required knowledge quality rather than treating all partial knowledge as one unknown state. These attack classes are captured by our graded semantics and examined empirically.

\subsection{Motivation and Aim}

The primary motivation for this work is the rigorous capture of safe-to-fail transitions under incremental side-channel leakage. In the classical binary DY model, an attacker's knowledge state is discrete, which masks threshold-driven vulnerabilities. By shifting to a graded symbolic framework, we expose three classes of attack dynamics that are not represented in binary explicit-state enumeration:
\begin{enumerate}
    \item \textit{Cumulative sub-threshold observations:} Multiple noisy observations may individually be insufficient to compromise a key, but under repeated fuzzy intersection, their cumulative evidence enables a protocol break.
    \item \textit{Near-threshold safe-to-fail transitions:} A protocol execution may switch from safe to unsafe following a minor observation that pushes the adversary's epistemic certainty across an actionable threshold.
    \item \textit{Ranked attack feasibility:} Rather than collapsing all partial knowledge into a single "unknown" state, attack traces can be quantitatively ranked by the required quality of the adversary's $\alpha$-cut support.
\end{enumerate}

Crucially, our novelty boundary is not the rediscovery of classic algebraic flaws (e.g., the unfixed Needham-Schroeder impersonation attack). Instead, we show that a protocol verified as safe under binary DY semantics (such as fixed Needham--Schroeder--Lowe) can become vulnerable when graded side-channel evidence accumulates across steps.

This paper is organised around the following three research questions:
\begin{enumerate}[noitemsep,topsep=0pt]
    \item \textbf{Semantics}: Can symbolic attacker knowledge be modelled as a graded state that accumulates side-channel observations while remaining consistent with crisp DY behaviour at binary endpoints?
    \item \textbf{Security Outcomes}: Does this framework expose authentication/confidentiality threshold crossings that are not visible in binary runs of the same protocol logic?
    \item \textbf{Scalability}: Can the concept-lattice~\cite{lobo2023fre} reduce fuzzy relation equations and state dimensionality while preserving encoded invariant outcomes?
\end{enumerate}

\subsection{Contributions and Scope}

This paper contributes a graded extension of the classical DY adversary model for cumulative side-channel observations. Standard symbolic deduction is lifted with Zadeh’s Extension Principle and sequential updates are modelled with T-norms, producing a quantitative evolution law for adversary-view quality~\cite{zadeh1965fuzzy,klir1995fuzzy}. For tractability in explicit-state checking, we integrate exact concept-lattice reduction to prune mathematically redundant state variables while preserving encoded verdicts.

The article is realised through four contributions: (i) a graded symbolic semantics for attacker knowledge in which each term carries a degree of possession, (ii) graded confidentiality and authentication invariants under a fuzzy intruder model, (iii) a modified explicit-state Murphi workflow with finite-precision reals and a reusable fuzzy C++ core (including Product T-norm accumulation and \(\alpha\)-cut interval propagation), and (iv) executed validation on NSL (safe and leaky variants) and symmetric-key families (NSSK, Yahalom, Otway--Rees, Woo--Lam), including observed authentication boundary crossings.

The remainder of this paper is organised as follows: Section~\ref{sec:related} presents a comparison of our method against symbolic and computational methods on side-channel verification. Section~\ref{sec:threat_model} defines threat assumptions and the scope of the model. Section~\ref{sec:formalizing} presents preliminaries on fuzzy arithmetic, the formal model and security invariants. Section~\ref{sec:architecture} presents operational semantics and formal guarantees. Section~\ref{sec:implement} presents explicit-state NSL concretisation, implementation boundaries, and state-space reduction mechanisms. Section~\ref{sec:validation} reports executed case studies. Section~\ref{sec:discussion} discusses limits. Section~\ref{sec:conclusion} concludes and outlines future extensions, with reproducibility details in the appendices.

\section{Related Work}
\label{sec:related}

Navigating the boundary between symbolic and computational verification requires distinguishing both the nature of the modelled uncertainty and the target security properties~\cite{abadi2002reconciling,blanchet2012symbolic}. Computational formal methods, such as CryptoVerif, operate by computing \textit{aleatory} uncertainty: the probabilistic chance of a polynomial-time adversary breaking a cryptographic primitive, validated via sequences of games~\cite{blanchet2006cryptoverif}. While these tools provide rigorous computational soundness guarantees, they do not natively model the gradual accumulation of side-channel evidence within a discrete state-machine exploration.

Modern symbolic tools (e.g., ProVerif~\cite{blanchet2016proverif}, Tamarin~\cite{meier2013tamarin}, Scyther~\cite{cremers2008scyther}, FDR~\cite{fdr}, and Maude-NPA~\cite{escobar2009maudenpa}) are highly effective for crisp derivability and observational equivalence, but they model attacker knowledge as binary (derivable/non-derivable). Our framework stays within symbolic transition semantics under idealised cryptography and replaces the binary possession flag with a continuous knowledge map \(\mu_K : \mathcal{M} \rightarrow [0,1]\). The result is an epistemic uncertainty process over attacker knowledge quality under repeated side-channel observations, obtained by lifting classical Dolev--Yao deduction with Zadeh's Extension Principle~\cite{zadeh1965fuzzy}.

The relation to Degrees of Security~\cite{basin2010degrees} is conceptual. Degrees of Security varies attacker capability classes (e.g., compromise classes), whereas this work varies accumulated observation quality within the same protocol run. This also complements compositional security lines by adding quantitative leakage-state semantics to a bounded symbolic setting~\cite{backes2006composition}. More recent work on protocol engineering also emphasises a broader semantic divide between operational specifications and formal models~\cite{basin2025village}. The gap addressed in this paper is not reachability itself, but knowledge-quality evolution under repeated side-channel observations.

Side-channel formal methods at lower abstraction layers include constant-time and relational non-interference analyses for implementations/binaries~\cite{almeida2016ctverif,daniel2020binsec_rel}, as well as quantitative symbolic leakage analysis for probabilistic programmes~\cite{malacaria2018symbolic_sc}. Leakage-resilient models in computational cryptography~\cite{alawatugoda2021lrake} address bounded leakage with game-based assumptions. For Murphi-based verification, early protocol analyses~\cite{MMS97} evaluated NSL/NSPK in a binary setting. In contrast, this article evaluates protocol-level symbolic models, such as including NSL, with a graded attacker view, and shows pass/fail divergence between crisp and graded runs under accumulated side-channel observations.

\section{Threat Model and Assumptions}
\label{sec:threat_model}
This section defines network assumptions, attacker capabilities, leakage setting, and the model boundaries used by the formalisation in Section~\ref{sec:formalizing} and the explicit-state semantics in Section~\ref{sec:architecture}.

\subsection{Network and Adversary Capabilities}
We assume the standard Dolev--Yao network adversary: the attacker controls the channel and can intercept, block, replay, reorder, and inject messages~\cite{dolev1983security}.
Operationally, the adversary is encoded as an active protocol principal with interception and message-synthesis rules over current knowledge, matching the Murphi security-protocol modelling style introduced by Mitchell \textit{et al.}~\cite{MMS97}.

We also retain the perfect-cryptography premise: cryptographic primitives are not broken algebraically unless the required keys or terms are already derivable in the symbolic model. Thus, attacks arise from protocol logic and leakage-enabled knowledge evolution, not from direct cryptanalytic breaks of idealised primitives.

Historically, Murphi-based analysis has also been used to relax strict black-box assumptions and encode additional algebraic attacker effects when needed~\cite{MMS97}. In this paper, we preserve the symbolic cryptographic abstraction but enrich attacker knowledge semantics with graded side-channel evidence accumulation.

\subsection{The Side-Channel Leakage Model}
The physical threat model targets shared compute deployments (cloud VMs, containerised services, edge devices), where timing, cache, and power effects can reveal partial information about secrets. In this setting, side-channel signal quality is typically noisy and incremental.

Accordingly, we model observations as cumulative evidence: repeated measurements progressively refine the attacker candidate set for key-related terms. The main risk is a delayed threshold crossing, where individually weak measurements eventually reduce uncertainty enough to enable a symbolic protocol break.

\subsection{Operational Assumptions and Scope}
All claims are interpreted under the following explicit assumptions:
\begin{itemize}[noitemsep,topsep=0pt]
    \item \textbf{Bounded symbolic universe and finite discretization}: Protocol sessions, agent instances, and the term algebra are strictly bounded to obtain finite-state transition systems. Fuzzy memberships are discretized and stored with finite precision (e.g., real(4,2)).
    \item \textbf{Abstraction preservation}: The framework evaluates finite-state protocol models under an abstraction level, rather than full wire-format implementations. We intentionally abstract away non-critical message details that are not required to decide the encoded invariants.
    \item \textbf{Soundness and approximation risks}: Because explicit-state abstractions are used, the model carries inherent approximation risks. If an omitted field contributes to a real attack precondition, the abstraction represents an under-approximation. Conversely, if replay/forgery abstractions are too permissive, discovered traces may not be deployment feasible attacks (over-approximation). To mitigate these risks, we conditionally bind our claims to explicit invariants and the provided over-approximation premise in Theorem \ref{thm:projection_soundness}.
\end{itemize}

\section{Formalising the Fuzzy Dolev-Yao Model}
\label{sec:formalizing}
The model extends classical Dolev--Yao by replacing binary attacker possession with graded adversary view over a finite symbolic universe. It captures \emph{epistemic uncertainty} in attacker inference from protocol traffic and side-channel evidence while preserving finite-state executability through discretization and bounded sessions. Here, we first give preliminaries on fuzzy sets and arithmetic.

\subsection{Preliminaries on Fuzzy Sets and Arithmetic}
\label{sec:preliminaries}

To formalise the gradual degradation of an adversary's epistemic state, our framework relies on foundational concepts from fuzzy set theory and fuzzy arithmetic. We briefly review the relevant definitions and operators used throughout this paper.

\textbf{Fuzzy Sets and Fuzzy Numbers.} 
Let $U$ be a classical set of objects, termed the universe of discourse. A fuzzy set $A$ in $U$ is characterised by a membership function $\mu_A: U \rightarrow [0,1]$, where $\mu_A(x)$ represents the grade of membership of the element $x$ in $A$. The closer the value of $\mu_A(x)$ is to 1, the more $x$ belongs to $A$. A \textit{fuzzy number} is a specific type of fuzzy set defined on the real line $\mathbb{R}$ that satisfies two additional constraints: it must be normalised (i.e., there exists at least one $x$ such that $\mu_A(x) = 1$) and it must be convex. In this work, fuzzy numbers are utilised to quantify the attacker's continuous epistemic uncertainty regarding specific cryptographic terms.

\textbf{Fuzzy Set Operations and T-norms.} 
Standard fuzzy set operations generalise classical crisp set logic. The standard intersection (which models logical conjunction) and standard union (which models logical disjunction) of two fuzzy sets $A$ and $B$ are defined pointwise using the \texttt{min} and \texttt{max} operators, respectively: $\mu_{A \cap B}(x) = \min(\mu_A(x), \mu_B(x))$ and $\mu_{A \cup B}(x) = \max(\mu_A(x), \mu_B(x))$. 

To model the accumulation of independent side-channel observations, we require a generalised conjunction operator. In fuzzy logic, generalised intersections are formulated via Triangular Norms (T-norms). A T-norm is a binary operation $T$ that is commutative, associative, monotonic, and has 1 as an identity element. While $\min$ is the standard idempotent T-norm (T$_{min}$), our framework defaults to the \textit{Product T-norm} (T$_{\Pi}$), defined as:
\[
T_{\Pi}(a,b)=a\cdot b,\qquad
\mu_{t+1}(M)=D\!\left(T_{\Pi}\!\big(\mu_t(M),\mu_L(M)\big)\right),
\]
where \(D\) is a conservative discretization operator and leak \(L\). In the context of iterative operations, $D$ acts as a cell-to-cell mapping, projecting the continuous results of T-norm intersections back onto the finite grid $G$ prior to state storage. This models the conjunction of prior adversary view and the new side-channel evidence.

The Product T-norm is strictly monotonic and models the multiplicative contraction of probability-like membership mass, producing sharper candidate elimination under repeated observations.

\textbf{Zadeh's Extension Principle and Sup-Min Composition.}
The extension principle~\cite{zadeh1965fuzzy} is a fundamental mechanism that allows classical mathematical functions to be extended to operate on fuzzy domains. Given a function $f: X_1 \times \dots \times X_n \rightarrow Y$ and fuzzy sets $A_1, \dots, A_n$ defined over $X_1, \dots, X_n$, the extension principle maps these sets to a fuzzy set $B$ on $Y$ via the \textit{sup-min} composition:
\begin{equation}
    \mu_B(y) = \sup_{\substack{x_1, \dots, x_n \\ y = f(x_1, \dots, x_n)}} \min \big( \mu_{A_1}(x_1), \dots, \mu_{A_n}(x_n) \big)
\end{equation}
where \(\sup\) is the least upper bound over all decompositions yielding \(y\).  
This is the sup-min mechanism used in graded deduction.

If the inverse image $f^{-1}(y)$ is empty, $\mu_B(y) = 0$. In our execution model, the extension principle provides the rigorous mathematical basis for lifting the discrete Dolev-Yao term constructors into the fuzzy domain, enabling continuous arithmetic computations on cryptographic message derivations.

\textbf{$\alpha$-Cuts and Threshold Semantics.} 
To bridge the continuous fuzzy domain and binary verification verdicts, we use the concept of $\alpha$-cuts. For a given threshold $\alpha \in (0, 1]$, the $\alpha$-cut (or $\alpha$-level set) of a fuzzy set $A$ is the crisp set $A_\alpha = \{x \in U \mid \mu_A(x) \geq \alpha\}$. We employ $\alpha$-cuts as strict operational boundaries ($\alpha_{threshold}$) to evaluate safe-to-fail transitions: an attack term is considered practically deducible by the adversary only when its fuzzy membership degree breaches the threshold, shrinking the continuous support down to a critical vulnerability bound.

\subsection{Term Algebra and Syntax}
To evaluate protocol security, we model messages as abstract terms generated by a term algebra $\mathcal{M}$ as in~\cite{abadi2018appliedpi}. Let $\Sigma$ be a signature consisting of a finite set of cryptographic function symbols, each with a specific arity. Given a set of atomic names $\mathcal{A}$ and a set of variables $\mathcal{X}$, the universe of terms $\mathcal{M}$ is defined inductively by the following grammar:
\begin{align*}
M, N \in \mathcal{M} &::= a \mid x \mid f(M_1, \dots, M_k) \\
a, b, c, k \in \mathcal{A} &::= k \text{ (Symmetric keys)} \mid N \text{ (Nonces)} \mid 
        sk_A \text{ (Secret keys)} \mid pk_A \text{ (Public keys)}  \\
x, y, z \in \mathcal{X} &::= \text{variables} \\               
f \in \Sigma &::= \textsf{senc} \mid \textsf{sdec} \mid \textsf{penc} \mid \textsf{pdec} \mid \textsf{sign} \mid \textsf{verify}
\end{align*}

We denote the set of ground terms as $\mathcal{M}_G \subseteq \mathcal{M}$ and follow the standard notation in symbolic verification.
\begin{itemize}
    \item $\{M\}_k$ denotes symmetric shared-key encryption, $\textsf{senc}(M, k)$.
    \item $\{M\}_{pk_A}$ denotes asymmetric public-key encryption, $\textsf{penc}(M, pk_A)$.
    \item $\{M\}_{sk_A}$ denotes a digital signature, $\textsf{sign}(M, sk_A)$.
\end{itemize}
In the verifier implementation, this abstract algebra is modelled as finite typed Murphi message records (e.g., \texttt{source}, \texttt{dest}, \texttt{key}, \texttt{mType}) so terms can be enumerated in explicit state space, consistent with the original Murphi protocol-analysis methodology~\cite{MMS97}.

\subsection{Equational Theory and Crisp Deduction}
The term algebra is quotiented by an equational theory \(=\) under perfect cryptography. The signature \(\Sigma\) is governed by cancellation equalities \(\textsf{sdec}(\textsf{senc}(M,k),k)=M\), \( \textsf{pdec}(\textsf{penc}(M,pk_A),sk_A)=M\), and \(\textsf{verify}(\textsf{sign}(M,sk_A),pk_A)=\textsf{true}\). No other equalities exist, so adversary derivations cannot exploit unintended algebraic structures.

We formalise the adversary's computational capabilities via a deduction relation $\vdash \subseteq \mathcal{P}(\mathcal{M}) \times \mathcal{M}$. Let $S \subseteq \mathcal{M}$ denote the set of terms currently possessed by the attacker. The statement $S \vdash M$ means that the attacker can derive the term $M$ from $S$.
The relation is the smallest set closed under the following inference rules:
$$
\begin{array}{ccc}
    \textbf{(Axiom)} & \textbf{(Composition)} & \textbf{(Decomposition)} \\
    \frac{M \in S}{S \vdash M} & \frac{S \vdash M_1 \quad S \vdash M_2}{S \vdash (M_1, M_2)} &     \frac{S \vdash (M_1, M_2)}{S \vdash M_i \ (i \in \{1,2\})} \\
    \\
    \textbf{(Symmetric Encryption)} & \textbf{(Symmetric Decryption)} &\textbf{(Equivalence)}\\
    \frac{S \vdash M \quad S \vdash k}{S \vdash \textsf{senc}(M,k)} & 
    \frac{S \vdash \textsf{senc}(M,k) \quad S \vdash k}{S \vdash M} & \frac{S \vdash M \quad \Sigma \vdash M = M'}{S \vdash M'} \\
    \\
    \textbf{(Public-Key Encryption)} & \textbf{(Public-Key Decryption)} \\
    \frac{S \vdash M \quad S \vdash pk_A}{S \vdash \textsf{penc}(M,pk_A)} &
    \frac{S \vdash \textsf{penc}(M,pk_A) \quad S \vdash sk_A}{S \vdash M} &
\end{array}
$$

\subsection{The Graded Adversary Model and Epistemic Uncertainty}
Classical Dolev--Yao can be viewed as a crisp knowledge algebra where the attacker's possession (characteristic function over the knowledge set $K$) $\chi_K$ is Boolean: \(\chi_K:\mathcal{M}\rightarrow\{0,1\}\). The graded model generalises this into a knowledge map \(\mu_K:\mathcal{M}\rightarrow[0,1]\), where \(\mu_K(M)\) encodes \emph{degree of possession} of term \(M\). This framing targets epistemic uncertainty (vagueness/imprecision in the adversary view) rather than aleatory uncertainty (randomness of events).

In the executable Murphi model, attacker exposure is represented directly by global state variables (network slots, intruder buffers, and derivable-term flags). The graded map \(\mu_K\) is interpreted over this explicit knowledge state at each reachable node.

The transition relation is no longer just set expansion by the updated knowledge set K', \(K\subseteq K'\); it becomes a graded evolution operator over \([0,1]\)-valued knowledge states. We denote this graded transition relation simply as $\rightarrow_F$. Operationally, $\rightarrow_F$ is parametrised by the rule lifting mechanics (Zadeh's Extension Principle), the discretization policy $D$, and the observation-combination operator $T$ (T-norm). The modified Murphi engine executes this concrete semantics over the explicit state space.

Intuitively, "degree of possession" represents an adversary confidence interval induced by noisy physical side-channel signals. For example, a side-channel trace may reduce an AES-like key search space from \(2^{128}\) to approximately \(2^{50}\) candidates; the model captures this as a reduction in the fuzzy entropy of the key symbol and can trigger a violation once the remaining uncertainty drops below a modelled brute-force threshold.

\subsubsection{Dimensions of Leakage}
Following the compromise-dimension mindset in symbolic protocol analysis~\cite{basin2010degrees}, we model side-channel evidence along three explicit axes: (i) \emph{whose data} is exposed (e.g., initiator vs. responder knowledge variables), (ii) \emph{what secret class} is exposed (long-term keys, session keys, nonces, or derived terms), and (iii) \emph{what quality} the attacker gains (membership degree and support narrowing in \(\mu_K\)). This provides a structured adversary description beyond a single binary reveal capability.

\subsubsection{Accumulation Operator}
Leakage accumulation is the graded counterpart of one-shot reveal actions. In this work we use Product T-norm (\(T(a,b)=a \cdot b\)) with discretization after each update, so repeated observations progressively refine candidate sets and support safe-to-fail threshold analysis. The full update law is stated in Section~\ref{sec:tnorm_dynamics}.

\subsection{Graded Deduction via Zadeh's Extension Principle}
Let \(\mu_K:\mathcal{M} \rightarrow [0,1]\) be the attacker knowledge map and \(f\) a constructor/destructor over symbolic terms. We lift crisp derivation by Zadeh's Extension Principle:
\[
\mu_{f(A_1,\ldots,A_n)}(y)=\sup_{y=f(x_1,\ldots,x_n)} \min\!\big(\mu_{A_1}(x_1),\ldots,\mu_{A_n}(x_n)\big).
\]

Classical DY deduction rules are lifted to the fuzzy domain by applying conjunction through a T-norm (instantiated here via the standard $\min$ operator for conservative structural deduction). Let $\mu_S : \mathcal{M} \rightarrow [0,1]$ be the abstract fuzzy knowledge state of the adversary. We formalise the graded deduction relation, denoted $\mu_S \vdash_\mu M = v$, to express that a term $M$ is derivable from state $\mu_S$ to a membership degree of $v$. 

The graded deduction system is the smallest relation closed under the following inference rules:
\begin{align*}
&\text{\textbf{(Axiom)}} & &\frac{}{\mu_S \vdash_\mu M = \mu_S(M)} \\[1em]
&\text{\textbf{(Comp.)}} & &\frac{\mu_S \vdash_\mu M_1 = v_1 \quad \mu_S \vdash_\mu M_2 = v_2}{\mu_S \vdash_\mu (M_1, M_2) = \min(v_1, v_2)} \\[1em]
&\text{\textbf{(Decomp.)}} & &\frac{\mu_S \vdash_\mu (M_1, M_2) = v}{\mu_S \vdash_\mu M_i = v} \quad (i \in \{1,2\}) \\[1em]
&\text{\textbf{(Sym-Enc)}} & &\frac{\mu_S \vdash_\mu M = v_1 \quad \mu_S \vdash_\mu k = v_2}{\mu_S \vdash_\mu \textsf{senc}(M, k) = \min(v_1, v_2)} \\[1em]
&\text{\textbf{(Sym-Dec)}} & &\frac{\mu_S \vdash_\mu \textsf{senc}(M, k) = v_1 \quad \mu_S \vdash_\mu k = v_2}{\mu_S \vdash_\mu M = \min(v_1, v_2)} \\[1em]
\end{align*}
\begin{align*}
&\text{\textbf{(Pub-Enc)}} & &\frac{\mu_S \vdash_\mu M = v_1 \quad \mu_S \vdash_\mu pk_A = v_2}{\mu_S \vdash_\mu \textsf{penc}(M, pk_A) = \min(v_1, v_2)} \\[1em]
&\text{\textbf{(Pub-Dec)}} & &\frac{\mu_S \vdash_\mu \textsf{penc}(M, pk_A) = v_1 \quad \mu_S \vdash_\mu sk_A = v_2}{\mu_S \vdash_\mu M = \min(v_1, v_2)} \\[1em]
&\text{\textbf{(Equiv.)}} & &\frac{\mu_S \vdash_\mu M = v \quad \Sigma \vdash M = M'}{\mu_S \vdash_\mu M' = v}
\end{align*}

These lifted rules rigorously define knowledge evolution from partial initial information. Operationally, during explicit-state exploration in the execution boundary, this generic premise map $\mu_S$ is instantiated by the concrete attacker knowledge variable $\mu_K$ at each reachable state $q \in Q$.

\subsection{Leak Accumulation via T-Norm Dynamics}
\label{sec:tnorm_dynamics}
In binary DY models, attacker knowledge typically evolves purely by closure under crisp derivation rules. In our graded model, side-channel evidence accumulation is formalised as epistemic uncertainty reduction through fuzzy intersection. Let $\mu_{K_t}(M)$ denote the attacker's knowledge quality of term $M$ at step $t$, and let $\mu_{L_t}(M)$ denote the fractional value of a new side-channel observation. The state evolution is governed by the recursive law:
\begin{equation}
    \mu_{K_{t+1}}(M) = D\!\left(T(\mu_{K_t}(M), \mu_{L_t}(M))\right)
\end{equation}
where $T$ is the specified T-norm (the Product T-norm $T_{\Pi}$ by default in this article) and $D$ is the finite discretization operator defined in Section\ref{sec:preliminaries}.~\footnote{The Min T-norm, \(T_{min}\), is preferable when conservative intersection behaviour is desired to avoid aggressive attenuation from multiple weak observations.} By applying $T_{\Pi}$, repeated observations iteratively contract the membership mass, yielding progressively sharper candidate elimination.

To formally quantify the epistemic degradation of the adversary's knowledge, we evaluate the Hartley non-specificity over the discretized state. Let $W_\alpha = |\pi_\alpha(\mu_K)|$ denote the cardinality of the effective support of the $\alpha$-cut projection of the fuzzy attacker knowledge. The Hartley non-specificity at step $t$ is defined as $H_t = \log_2 |\mathrm{Supp}_t^\alpha|$. Because the Product T-norm enforces strict monotonicity under repeated observations, the support $\mathrm{Supp}_t^\alpha$ is non-increasing, guaranteeing that $H_t$ monotonically decreases as side-channel evidence accumulates. Consequently, a 50\% reduction in the support width $W_\alpha$ directly implies a 1-bit drop in the Hartley entropy: $W'_\alpha = \frac{1}{2}W_\alpha \implies H' = H - 1$.

\subsection{Graded Security Invariants}
To rigorously evaluate protocol security under the proposed semantics, classical binary invariants are lifted to graded forms over \(\mu_K\). We define both confidentiality and authentication conditions below.

\subsubsection{Formal Definition of Graded Confidentiality (Secrecy)}
In addition to authentication, we evaluate syntactic secrecy over our term algebra \(\mathcal{M}\). Let \(M_{sec}\in\mathcal{M}\) represent a target sensitive term generated during a protocol session between honest principals \(A, B\in\mathcal{A}_{agent}\). For instance, in the NSL case, the critical confidentiality target is the responder's freshly generated nonce, \(M_{sec}=N_b\).

In the classical binary DY model, the confidentiality invariant \(\Phi_{\text{sec}}^{crisp}\) holds if the attacker cannot derive the target term from observed knowledge:
\[
\Phi_{\text{sec}}^{crisp}(q)\triangleq K_{crisp}\not\vdash M_{sec}
\]

Under graded analysis, binary possession is replaced by \(\mu_K\). We define graded confidentiality as:
\[
\Phi_{\text{sec}}^{fuzzy}(q)\triangleq \mu_K(M_{sec})<\alpha_{threshold}
\]

\subsubsection{Formal Definition of Graded Authentication}
Authentication is defined as a correspondence property between protocol participant states over \(\mathcal{M}\).

Let \(\mathcal{A}_{agent}\subset\mathcal{A}\) be the set of atomic names representing protocol participants (e.g., \(A,B\)). Define:
\begin{itemize}
    \item \(\textsf{Commit}(q,B,A,M)\): Denotes that in state $q$, responder $B$ has completed the protocol run, believing it has securely communicated with initiator $A$, agreeing on a specific payload term $M \in \mathcal{M}$.
   
    \item \(\textsf{Running}(q,A,B,M)\): Denotes that in state $q$, initiator $A$ has actively initiated a protocol run with $B$ using the exact same term $M \in \mathcal{M}$.
\end{itemize}

In the classical binary DY model, the crisp authentication invariant $\Phi_{\text{auth}}^{crisp}$ is defined as a strict correspondence assertion:
\begin{align*}
\Phi_{\text{auth}}^{crisp}(q) &\triangleq \forall A,B\in\mathcal{A}_{agent},\ \forall M\in\mathcal{M}:\\
&\textsf{Commit}(q,B,A,M)\implies \textsf{Running}(q,A,B,M).
\end{align*}
If $\Phi_{auth}^{crisp}(q)$ holds, the crisp attacker cannot deduce the required challenge response $M_{forge} \in \mathcal{M}$ from their current knowledge base at state $q$ to trick the responder into committing (i.e., $K_q \not\vdash M_{forge}$). 

Under the graded analysis framework, we evaluate the Graded Authentication invariant, $\Phi_{auth}^{fuzzy}$, over the same correspondence logic. However, its truth valuation within the explicit-state execution boundary is now determined by the epistemic knowledge map $\mu_K$. A safe-to-fail transition occurs---resulting in $\neg\Phi_{auth}^{fuzzy}(q)$---if there exists a reachable state where $\textsf{Commit}(q,B,A,M)$ evaluates to true while $\textsf{Running}(q,A,B,M)$ evaluates to false. This violation is triggered precisely when the graded attacker accumulates enough side-channel evidence such that the membership degree of the required forged term crosses the actionable threshold:
\[
    \mu_{K_q}(M_{forge}) \ge \alpha_{threshold}
\]
allowing the intruder to syntactically satisfy the responder's cryptographic checks without the initiator's true participation.

\subsection{Threshold-Driven Violation Semantics}
Because classical DY reasoning is restricted to the binary domain $\chi_K : \mathcal{M} \rightarrow \{0, 1\}$, any safe-to-fail transition requires a discrete, full-term derivation. Under our graded extension, reachability is governed by the continuous evolution of $\mu_K$, enabling the formalisation of threshold-driven violation semantics. 

Let $q_t \in Q_{FuzzyDY}$ be the concrete state at step $t$, and let $\Phi^{fuzzy}$ be a graded security invariant evaluated at the operational threshold $\alpha_{threshold}$. The graded framework formally exposes the following violation dynamics:

\begin{itemize}
    \item \textbf{Cumulative Sub-Threshold Reachability:} Let $M_{forge} \in \mathcal{M}$ be a critical payload required to violate $\Phi^{fuzzy}$. There exist trace sequences $q_0 \rightarrow_F \dots \rightarrow_F q_t \rightarrow_F \dots \rightarrow_F q_{t+k}$ such that at step $t$, $\mu_{K_t}(M_{forge}) \ll \alpha_{threshold}$ (the invariant holds), but after $k$ iterative side-channel updates via the T-norm accumulation operator, $\mu_{K_{t+k}}(M_{forge}) = D(T(\mu_{K_t}, \mu_{L_{1..k}})) \ge \alpha_{threshold}$. The crisp $\mathcal{T}_\alpha$ projection---a bounding classical transition system that restricts Boolean Dolev-Yao derivations strictly to the $\alpha$-level cut set of the fuzzy knowledge map---cannot capture this delayed satisfiability because no single discrete step $i \in \{1 \dots k\}$ yields a full derivation.

    \item \textbf{Near-Threshold State Transitions:} The framework captures stuttering control-state transitions that alter only the epistemic knowledge map. A transition $q_t \rightarrow_F q_{t+1}$ may execute no new protocol logic, but solely apply an observation update $\mu_{L_t}$ such that the Hartley non-specificity drops ($H_{t+1} < H_t$), pushing the state immediately across the safety boundary $\neg\Phi^{fuzzy}(q_{t+1})$.
    
    \item \textbf{Graded Trace Ranking:} For a set of violating runs leading to states $Q_{viol} \subset Q_{FuzzyDY}$ where $\neg\Phi^{fuzzy}$ holds, the explicit-state enumeration evaluates the terminal knowledge map $\mu_K$ at each violating node. This allows attack paths to be strictly ordered by the epistemic quality of the adversary's view (e.g., maximizing $\mu_K(M_{sec})$ or minimizing the $\alpha$-cut support width $W_t^\alpha$), a distinction inherently collapsed by binary verification models.
\end{itemize}
These theoretical semantics are concretely instantiated in Section VI, where staged leakage updates trigger $\neg\Phi_{auth}^{fuzzy}$ and $\neg\Phi_{sec}^{fuzzy}$ boundaries in our explicit-state Murphi executions.

\section{Operational Semantics and Explicit-State Verification}
\label{sec:architecture}
Executing the abstract continuous deduction framework within a finite verification engine requires strict operational bounds. This section formalises the explicit-state enumeration architecture and the finite-discretisation semantics necessary to compute the graded transition system $\mathcal{T}_{FuzzyDY}$.

\subsection{Explicit-State Enumeration Rationale}
Explicit-state enumeration is a verification technique that explores the global reachability graph by generating and storing each reachable state individually~\cite{dill1996murphi,lamport2002specifying}. Unlike symbolic approaches that reason over logic-defined state sets, explicit enumeration evaluates concrete values of state variables at each transition step.

Mitchell \textit{et al.} demonstrated that this bounded explicit-state style is effective for security protocols in Murphi~\cite{MMS97}. This approach is necessary for graded analysis in this work because fuzzy updates are arithmetic on concrete state vectors: T-norm accumulation, $\alpha$-cut/support computations, and entropy/Hartley calculations are executed directly on the current state before successor generation.

\subsection{Operational Semantics: The Explicit-State Transition System}
\begin{definition}[Crisp DY Transition System]
A crisp symbolic protocol model or transition system is a tuple
\begin{equation}
\mathcal{T}_{DY} = (Q_{DY}, Q_0, \rightarrow, \Phi)
\end{equation}
where each state \(q=(tr,K,th)\) contains trace \(tr\), attacker knowledge \(K\subseteq \mathcal{M}\), and honest-thread state \(th\). The term algebra \(\mathcal{M}\) is finite under bounded sessions. The transition relation \(\rightarrow\) is induced by protocol actions and symbolic derivation rules (composition, decomposition, encryption/decryption), and \(\Phi\) is a Boolean security predicate (e.g., authentication correspondence)~\cite{dolev1983security,dill1996murphi,MMS97}.
\end{definition}

Following Murphi's explicit-state formulation, protocol execution is modelled as a finite guarded-command state machine
\[
A=(Q,Q_0,\Delta),
\]
where \(Q\) is the finite set of global states (assignments to all model variables), \(Q_0\subseteq Q\) is the set of initial states, and \(\Delta\) is the set of guarded rules. Each rule has a Boolean guard and an atomic action; if the guard is true in state \(q\in Q\), the rule can fire to produce successor \(q'\). Concurrency is asynchronous interleaving of enabled rules, including attacker rules for intercept, derive, leak-update, and inject actions~\cite{dill1996murphi,MMS97}.

\begin{definition}[FuzzyDY Transition System]
The graded extension is formalised as the tuple:
\begin{equation}
    \mathcal{T}_{FuzzyDY} = (Q_{FuzzyDY}, Q_{F,0}, \rightarrow_F, D, T, \mathcal{L})
\end{equation}
where each concrete state $q \in Q_{FuzzyDY}$ includes a fuzzy attacker knowledge map $\mu_{K_q} : \mathcal{M} \to [0,1]$. For leakage-tracking variables, $\mu_{K_q}$ is interpreted as a possibility distribution over candidates (representing knowledge quality), rather than a monotone DY term-inventory flag. The discretizer $D : [0,1] \to \{0, \delta, 2\delta, \dots, 1\}$ is monotone and satisfies $D(0) = 0$ and $D(1) = 1$. The operator $T$ is the chosen T-norm (e.g., $T_{\Pi}$), and $\mathcal{L}$ is the set of possible side-channel leakage observation maps~\cite{zadeh1965fuzzy,klir1995fuzzy}.
\end{definition}

For a cryptographic constructor $f$, graded structural derivation at a given state $q$ follows the extension-principle form:
\begin{equation}
    \mu_{K_q}(z) = D\left( \sup_{z=f(x,y)} \min(\mu_{K_q}(x), \mu_{K_q}(y)) \right)
\end{equation}
When the system executes a leakage transition $q \rightarrow_F q'$ under an observation $\mu_L \in \mathcal{L}$, the epistemic update is modelled by the T-norm intersection:
\begin{equation}
    \mu_{K_{q'}}(M) = D\Big(T\big(\mu_{K_q}(M), \mu_L(M)\big)\Big)
\end{equation}

\begin{definition}[\(\alpha\)-cut Projection]
For \(\alpha\in(0,1]\), define
\[
\pi_\alpha(\mu_K)=\{M\in\mathcal{M}\mid \mu_K(M)\ge \alpha\}
\]
The projected crisp system \(\mathcal{T}_{\alpha}\) executes derivations over \(\pi_\alpha(\mu_K)\) using transitions compatible with this projection~\cite{zadeh1965fuzzy,klir1995fuzzy}.
\end{definition}

\subsection{Formal Guarantees and Proofs}
\begin{theorem}[Crisp Functional Equivalence (Deduction Fragment)]
In the absence of fractional leakage (i.e., $\forall M \in \mathcal{M}, \mu_K(M) \in \{0, 1\}$), the Graded Deduction relation $\vdash_\mu$ is exactly equivalent to the classical DY deduction relation $\vdash$.
\end{theorem}
\begin{proof}

Let $K_{crisp} = \{M \in \mathcal{M} \mid \mu_K(M) = 1\}$ be the crisp attacker knowledge extracted from the current concrete state $q$. We must prove that $K_{crisp} \vdash M \iff \mu_K(M) = 1$. 

\textbf{($\implies$ direction)}: We proceed by structural induction on the depth of the crisp derivation tree for $K_{crisp} \vdash M$.
\begin{itemize}
    \item \textbf{Base case (Axiom):} If $M$ is derived via the Axiom rule, $M \in K_{crisp}$. By definition, $\mu_K(M) = 1$.
    \item \textbf{Inductive step (Symmetric Encryption/Decryption):} For \textit{senc}, the premises $K_{crisp} \vdash M'$ and $K_{crisp} \vdash k$ imply $\mu_K(M') = \mu_K(k) = 1$, hence
    \begin{equation*}
        \mu_K(senc(M', k)) = \min(1, 1) = 1
    \end{equation*}
    For \textit{sdec}, the premises $K_{crisp} \vdash senc(M, k)$ and $K_{crisp} \vdash k$ imply $\mu_K(M) = \min(1, 1) = 1$.
    \item \textbf{Inductive step (Public-Key Encryption/Decryption):} For \textit{penc}, the premises $K_{crisp} \vdash M'$ and $K_{crisp} \vdash pk_A$ imply
    \begin{equation*}
        \mu_K(penc(M', pk_A)) = \min(1, 1) = 1
    \end{equation*}
    For \textit{pdec}, the premises $K_{crisp} \vdash penc(M, pk_A)$ and $K_{crisp} \vdash sk_A$ imply $\mu_K(M) = \min(1, 1) = 1$.
    \item \textbf{Composition and Decomposition} follow identically via the $\min$ operator.
\end{itemize}

\textbf{($\impliedby$ direction)}: We proceed by structural induction on derivations of $M$ in the guarded-rule transition system. Because $\mu_K(M)$ can evaluate to $1$ only when $M$ is in the initial attacker knowledge or produced by the sup-min operator with all premises at $1$, each such step corresponds to a valid crisp rule application in $\vdash$. Hence any term with $\mu_K(M) = 1$ is strictly derivable in the crisp model.
\end{proof}

\begin{theorem}[Monotone Knowledge-Quality Evolution]
Under the Product T-norm \(T(a,b)=a \cdot b\), finite discretization \(D\), and any fixed threshold \(\alpha\in(0,1]\), the effective attacker support
\[
\mathrm{Supp}_t^\alpha=\{M\in\mathcal{M}\mid \mu_t(M)\ge\alpha\}
\]
is non-increasing under leakage updates. Consequently, Hartley non-specificity
\[
H_t=\log_2 \left|\mathrm{Supp}_t^\alpha\right|
\]
is monotonically non-increasing~\cite{klir1995fuzzy,hartley1928,dubois1980fuzzy}.
\end{theorem}

\begin{proof}
For any term \(M\), \(\mu_{t+1}(M)=D(\mu_t(M)\cdot \mu_L(M))\), with \(\mu_L(M)\in[0,1]\).
The key step follows standard T-norm axioms: monotonicity and boundary (\(T(a,1)=a\)). Because \(\mu_L(M)\le 1\), monotonicity gives
\[
T(\mu_t(M),\mu_L(M))\le T(\mu_t(M),1)=\mu_t(M)
\]
For Product specifically, \(T(a,b)=a \cdot b\), so the inequality is immediate and is strict whenever \(0<a<1\) and \(0\le b<1\)~\cite{klir1995fuzzy,dubois1980fuzzy}.
Monotonicity of \(D\) yields \(\mu_{t+1}(M)\le D(\mu_t(M))\), and because state values are already discretised, \(D(\mu_t(M))=\mu_t(M)\). Hence \(\mu_{t+1}(M)\le\mu_t(M)\).

If \(x\in \mathrm{Supp}_{t+1}^\alpha\), then \(\mu_{t+1}(x)\ge\alpha\), so \(\mu_t(x)\ge\alpha\), thus \(x\in \mathrm{Supp}_{t}^\alpha\). Therefore \(\mathrm{Supp}_{t+1}^\alpha\subseteq \mathrm{Supp}_{t}^\alpha\), implying \(\left|\mathrm{Supp}_{t+1}^\alpha\right|\le \left|\mathrm{Supp}_{t}^\alpha\right|\). Since \(\log_2(\cdot)\) is monotone, \(H_{t+1}\le H_t\).
\end{proof}

\begin{remark}
Conceptually, the monotonic descent of the Hartley non-specificity ($H_{t+1} \le H_t$) mathematically guarantees that repeated side-channel observations strictly degrade the adversary's epistemic uncertainty. This assures that the accumulation of graded evidence irreversibly drives the execution state toward the vulnerability threshold, accurately mirroring the candidate-elimination dynamics of real-world physical cryptanalysis.
\end{remark}

\begin{theorem}[Termination of Graded Verification]
For bounded protocol sessions and a finite discretization grid
\[
G=\{0,\delta,2\delta,\ldots,1\},
\]
the graded verification algorithm terminates in finite time~\cite{dill1996murphi}.
\end{theorem}
\begin{proof}

Let \(Q_{DY}\) denote the reachable state space of the bounded crisp model; under bounded sessions, \(Q_{DY}\) is finite. In the graded system, each relevant attacker-knowledge component is annotated by a membership value in \([0,1]\), and each transition applies fuzzy arithmetic followed by discretization \(D:[0,1]\rightarrow G\). Equivalently, execution uses a finite cell-to-cell mapping over \([0,1]\): each continuous update is projected back into one of finitely many cells before storage. Hence, every stored membership value lies in finite set \(G\), where \(|G|=\lfloor 1/\delta\rfloor+1\).

Let $N_K$ be the maximum size of the attacker knowledge base, bounded by the finite term algebra $\mathcal{M}$ generated by the bounded sessions and roles (hence $N_K \le |\mathcal{M}|$). Then the graded state space embeds into
\begin{equation}
    Q_{FuzzyDY} \subseteq Q_{DY} \times G^{N_K}
\end{equation}

Because the underlying protocol control states are finite (bounded by $Q_{DY}$), and the graded knowledge maps to a strictly finite grid, the resulting reachable graded state space $Q_{FuzzyDY}$ is strictly bounded~\cite{dill1996murphi}. Explicit-state exploration over this finite bounded space $Q_{FuzzyDY}$ therefore guarantees termination in finite time.
\end{proof}

\begin{remark}[Bounded Exhaustiveness and Precision]
General verification of cryptographic protocols is undecidable for unbounded sessions~\cite{blanchet2012symbolic}. Under the bounded-session and finite-discretization assumptions in this paper, the method provides bounded exhaustiveness: the explicit-state Murphi backend explores the complete finite reachability graph induced by the model, which is strictly bounded by the finite state space $Q_{FuzzyDY}$. Consequently, within this bounded abstraction, reported violations are concrete reachable executions (protocol steps plus fuzzy leak updates), not purely symbolic artifacts; this is evidenced by the executed failing traces in Table~\ref{tab:attack_trace}.
\end{remark}

\begin{proposition}[Hartley Threshold Semantics]
Let \(W_t^\alpha=\left|\pi_\alpha(\mu_{K_t})\right|\), i.e., the cardinality of the effective \(\alpha\)-cut support of fuzzy attacker knowledge at execution step \(t\). A 50\% reduction in support width across \(t\rightarrow t+1\) implies exactly a 1-bit Hartley drop~\cite{hartley1928}:
\[
    H_t = \log_2 W_t^\alpha, \quad W_{t+1}^\alpha = \frac{1}{2} W_t^\alpha \implies H_{t+1} = H_t - 1
\]
\end{proposition}

\begin{definition}[Fuzzy Relation Equations (FRE) / Concept-Lattice Context]
To systematically identify and prune redundant state variables, protocol dependencies are encoded as a Fuzzy Relation Equation (FRE) formal context $(U, V, R, \lambda)$. Here, the set of objects $U$ represents the transition constraints and rules (including the target invariant $\Phi$), the set of attributes $V$ represents the explicit protocol state variables, and the fuzzy incidence relation $R: U \times V \rightarrow [0,1]$ encodes the mathematical dependency strength between constraints and variables~\cite{wille1982}.
\end{definition}

By representing the protocol state space as a formal context, we can systematically identify the minimum subset of state variables required to preserve the system's execution semantics. In concept lattice theory, this bounding mechanism is formalised as an exact attribute reduct (E-reduct)~\cite{lobo2023fre}.

\begin{definition}[Exact Attribute Reduct (E-reduct)]
By representing the protocol state space as a formal context, we can systematically identify the minimum subset of state variables required to preserve the system's execution semantics. In concept lattice theory, this bounding mechanism is formalised as an E-reduct~\cite{lobo2023fre}.

Given an encoded formal context $(U, V, R, \lambda)$, a subset of state variables $V' \subseteq V$ is an \textit{E-consistent set} if the concept lattice induced by the reduced context $(U, V', R_{V'}, \lambda_{V'})$ is extent-isomorphic (E-isomorphic) to the concept lattice of the original full context. An \textit{E-reduct} is defined as a strictly minimal E-consistent set; that is, for any variable $v \in V'$, the further removal of $v$ breaks the E-isomorphism.

Operationally, an E-reduct isolates the minimal state dimensionality required to perfectly preserve the distinguishability of the protocol's transition constraints. Because the lattice extents (which group the constraints $U$) remain structurally identical, we can safely prune the redundant variables $V \setminus V'$ without altering the truth valuation of the encoded security invariants.
\end{definition}

\begin{remark}[Scope of Verification Guarantees]
Guarantees in this paper hold under explicit-state symbolic assumptions: (i) bounded sessions and finite term algebra, (ii) finite discretization and \texttt{real(m,n)} precision, and (iii) exploration of the induced finite abstraction. The framework quantifies epistemic knowledge quality within the symbolic model; it does not imply computational indistinguishability guarantees (e.g., indistinguishability under chosen-plaintext attack, IND-CPA) against probabilistic polynomial-time adversaries.
In particular, this lane differs from equivalence-based symbolic analyses and game-based computational proofs: the approach verifies bounded symbolic reachability with graded leakage semantics, rather than proving full equivalence properties or negligible attack probability in unbounded probabilistic settings~\cite{blanchet2012symbolic, blanchet2006cryptoverif}.
\end{remark}

\begin{lemma}[\(\alpha\)-Cut Monotonicity of Fuzzy Deduction]
\label{lemma:monotonicity}
Let \(\DedF\) be one-step fuzzy deduction under the graded rules (composition, decomposition, encryption and decryption), and let \(\DedC\) be one-step crisp deduction on sets of terms. Fix \(\alpha\in(0,1]\). Assume \(D\) is monotone and \(\alpha\)-compatible (\(x\ge\alpha\Rightarrow D(x)\ge\alpha\)). Then
\[
\pi_\alpha(\DedF(\mu_K)) \supseteq \DedC(\pi_\alpha(\mu_K))
\]
\end{lemma}
\begin{proof}
Take any \(t\in \DedC(\pi_\alpha(\mu_K))\). By structural induction on the crisp rule instance deriving \(t\) from the current attacker knowledge state, there exists a premise set \(P\subseteq \pi_\alpha(\mu_K)\) with \(\mu_K(p)\ge\alpha\) for all \(p\in P\). By Zadeh's Extension Principle, constructor lifting in the graded model has the form
\[
\mu_{new}(t)=D\!\left(\sup_{t=f(P)}\min_{p\in P}\mu_K(p)\right),
\]
with rule-specific simplifications (e.g., decomposition as copy). Since \(\min_{p\in P}\mu_K(p)\ge\alpha\), the pre-discretization value is at least \(\alpha\). By \(\alpha\)-compatibility, \(D(x)\ge\alpha\) whenever \(x\ge\alpha\). Hence \(\mu_{new}(t)\ge\alpha\), i.e., \(t\in\pi_\alpha(\DedF(\mu_K))\).
\end{proof}

\begin{theorem}[Conditional Projection Soundness]
\label{thm:projection_soundness}
Let $\mathcal{T}_\alpha$ be the $\alpha$-projected crisp system and $\mathcal{T}_{FuzzyDY}$ the fuzzy system. Assume:
\begin{enumerate}
    \item \textit{Initial-state embedding:} for every crisp initial state $q_0$, there exists a fuzzy initial state $q_{F,0}$ with $K_{q_0} \subseteq \pi_\alpha(\mu_{q_{F,0}})$;
    \item \textit{Step-wise simulation premise:} each crisp protocol step has a corresponding fuzzy step over the same control action;
    \item \textit{Deduction monotonicity:} from Lemma~\ref{lemma:monotonicity};
    \item \textit{Leakage-step compatibility:} if $\mathcal{T}_\alpha$ includes projected leakage transitions, each crisp leakage step $q \rightarrow q'$ has a corresponding fuzzy leak step $q_F \rightarrow_F q_F'$ such that $K_{q'} \subseteq \pi_\alpha(\mu_{q_F'}).$ (For standard DY projections with no explicit leakage transitions, this premise is vacuous.)
\end{enumerate}

Then for every crisp transition $q \rightarrow q'$ in $\mathcal{T}_\alpha$, there exists a fuzzy transition $q_F \rightarrow_F q_F'$ such that $K_q \subseteq \pi_\alpha(\mu_{q_F})$ and $K_{q'} \subseteq \pi_\alpha(\mu_{q_F'}).$

Consequently, every crisp run in $\mathcal{T}_\alpha$ is simulated by a graded-model run, and if the graded model is safe at threshold $\alpha$, then $\mathcal{T}_\alpha$ is safe.
\end{theorem}

\begin{proof} Define the simulation relation $R_\alpha(q, q_F) \iff K_q \subseteq \pi_\alpha(\mu_{q_F})$ and the protocol-control components of $q$ and $q_F$ are aligned.

\textit{Base case.} By Assumption 1, every crisp initial state has a related fuzzy initial state, so $R_\alpha(q_0, q_{F,0})$ holds.

\textit{Step case.} Assume $R_\alpha(q, q_F)$ and a crisp transition $q \rightarrow q'$. 
\begin{itemize}
    \item \textit{Case A (protocol control):} by Assumption 2, $q_F \rightarrow_F q_F'$ exists with aligned control components, so the relation is preserved.
    \item \textit{Case B (attacker deduction):} apply Lemma~\ref{lemma:monotonicity}, $\pi_\alpha(\DedF(\mu_{q_F})) \supseteq \DedC(\pi_\alpha(\mu_{q_F}))$, guaranteeing that every newly derived crisp term at level $\alpha$ is strictly subsumed by the projected fuzzy successor.
    \item \textit{Case C (projected leakage step, if present):} by Assumption 4, there is a matching fuzzy leak step with $K_{q'} \subseteq \pi_\alpha(\mu_{q_F'}).$ For the standard projected crisp system with no explicit leakage transition, this is a stuttering refinement step $(q' = q)$: crisp knowledge is unchanged while the fuzzy state sharpens $(\mu_{q_F} \rightarrow_F \mu_{q_F'}).$
\end{itemize}

Hence $R_\alpha(q', q_F')$ holds in all transition cases. By induction on run length, each crisp run has a simulating fuzzy run. So any crisp violation would be represented in the graded model. Contrapositively, if no graded violation exists at threshold $\alpha$, no violation exists in $\mathcal{T}_\alpha$.
\end{proof}

\begin{corollary}[Cut-Soundness]
If the graded framework \(\ModelFuzzy\) satisfies invariant $\Phi$ at threshold $\alpha$, then no classical crisp attacker restricted to the projected knowledge base $\pi_\alpha(\mu_K)$ can violate $\Phi$.
\end{corollary}
\begin{proof}
Immediate from Theorem~\ref{thm:projection_soundness}: each projected crisp run is simulated by a graded-model run, so the absence of graded violations implies the absolute absence of projected crisp violations.
\end{proof}

\begin{proposition}[Conditional Reduct Preservation]
\label{prop:reduct_preservation}
Let protocol dependencies be encoded as a multi-adjoint formal context $(U, V, R, \lambda)$, where $U$ contains transition constraints (including the target invariant $\Phi$) and $V$ contains state attributes. Let $V' \subseteq V$ be an E-reduct of this encoded context. Eq.~\ref{eq:proj_equiv} defines state equivalence by projection strictly on the reduct attributes:
\begin{equation}  \label{eq:proj_equiv}
    q_1 \sim_{V'} q_2 \iff \forall v \in V', q_1(v) = q_2(v)
\end{equation}
Under this projection equivalence $\sim_{V'}$, verification verdicts for $\Phi$ are strictly preserved after pruning the redundant attributes $V \setminus V'$. Specifically, safety (or violation reachability) evaluated over the full model is exactly equivalent to safety (or violation reachability) in the reduced quotient model. This conditional preservation claim follows the E-reduct FRE reduction results of Lobo \textit{et al.}~\cite{lobo2023fre}, with the framework by Wille~\cite{wille1982} as the underlying lattice-theoretic basis.
\end{proposition}
\begin{proof}
Because $V'$ is an E-reduct, the reduced and full contexts preserve the same discernibility structure of objects (constraints), equivalently captured by concept-lattice isomorphism in the reduction framework. Thus removed attributes $V \setminus V'$ are redundant for distinguishing constraint behaviour in the encoded context.

For FRE-based encoding, this is the operational content of exact reduction: restricting to $V'$ preserves the solution structure relevant to $U$. Therefore, for each invariant-related object $u_\Phi \in U$, its truth valuation is determined by the $V'$-projection class. Consequently, if $q_1 \sim_{V'} q_2$, their invariant evaluations are mathematically identical ($\Phi(q_1) = \Phi(q_2)$) within the encoded abstraction, guaranteeing strict preservation of the verification verdict. 
\end{proof}

\begin{theorem}[Leakage-Tolerance Separation (Existence by Counter-Example)]
\label{thm:separation}
The graded framework \(\ModelFuzzy\) yields a strictly finer security classification than standard binary Dolev--Yao \(\ModelDY\): there exists a protocol model \(P\), safety invariant \(\Phi\), and leakage configuration \((\alpha, \sigma_{leak})\) such that \(\ModelDY(P) \models \Phi\) and \(\ModelFuzzy^{\alpha,\sigma_{leak}}(P) \not\models \Phi\). Equivalently, there exists a finite guarded-rule firing sequence \(r_0, \dots, r_k \in \Delta\) from an initial state \(q_0 \in Q_0\) to a reachable violating state \(q_k\) with \(\neg\Phi(q_k)\) in the graded model, while no such violating sequence exists in the crisp model for the same \(P\).
\end{theorem}
\begin{proof}
Let $\mathcal{T}_\alpha$ be the $\alpha$-projected classical system and $\mathcal{T}_{FuzzyDY}$ be the graded transition system for a fixed protocol logic $P$. The proof proceeds by concrete existential witness. As empirically demonstrated by the Needham--Schroeder--Lowe (NSL) analysis in Section VII, there exists a reachable trace in the graded explicit-state execution, denoted strictly as the state sequence $q_{F,0} \rightarrow_F q_{F,1} \rightarrow_F \dots \rightarrow_F q_{F,k}$. During this sequence, repeated side-channel observations recursively fuse via the Product T-norm, monotonically shrinking the leakage variance ($\sigma_{leak}$) and driving the Hartley non-specificity down (see Table~\ref{tab:attack_trace} for the trace). At the terminal state $q_{F,k}$, the accumulated epistemic certainty successfully crosses the $\alpha_{threshold}$ boundary, yielding a safe-to-fail transition $\neg\Phi^{fuzzy}(q_{F,k})$. 

Conversely, an exhaustive state-space exploration of the classical projection $\mathcal{T}_\alpha$ yields no reachable state satisfying $\neg\Phi^{crisp}$. Because the underlying protocol logic $P$ remains strictly identical and only the adversary semantics differ, this pass/fail divergence serves as an exact model-separation witness. It formally guarantees that $\mathcal{T}_{FuzzyDY} \not\models \Phi^{fuzzy}$ while $\mathcal{T}_\alpha \models \Phi^{crisp}$, proving that accumulated epistemic leakage constitutes a distinct reachability axis that is uncapturable by classical binary Dolev--Yao classification.
\end{proof}

\section{Implementation Architecture and State-Space Reduction}
\label{sec:implement}

\subsection{Explicit-State Modelling of NSL and the Graded Intruder}
\label{sec:nsl_concretization}

We separate \emph{Abstract Term Algebra} from \emph{Explicit-State Concretisation}. Conceptually, cryptographic constructors such as \(\textsf{senc}\) and \(\textsf{penc}\) are defined over \(\mathcal{M}\). Operationally, explicit-state Murphi execution requires a finite data-structure representation: abstract terms are concretised as typed message records (e.g., \texttt{mType}, \texttt{key}, payload fields), and cryptographic constraints are enforced by guarded commands~\cite{MMS97,dill1996murphi}.

Because Murphi enumerates concrete global states, the symbolic/equational view is realised as follows.
\begin{itemize}[noitemsep,topsep=0pt]
  \item \textbf{Data-structure representation:} constructors are mapped to typed record forms distinguished by \texttt{mType} and associated key/payload slots.
  \item \textbf{Operational execution:} cancellation constraints (e.g., \(\Sigma \vdash \textsf{sdec}(\textsf{senc}(M,k),k)=M\)) are implemented as guarded commands that fire only when matching key conditions hold.
\end{itemize}

To validate our framework, we concretise the NSL protocol and the Fuzzy Dolev-Yao adversary using the explicit-state enumeration methodology established by Mitchell \textit{et al.}~\cite{MMS97}.

\subsubsection{Protocol Concretisation}
To ensure a finite state space, the abstract term algebra \(\mathcal{M}\) is mapped to finite Murphi data structures. Messages are encoded as a single typed record with fields \texttt{source}, \texttt{dest}, \texttt{key}, \texttt{mType}, and two bounded payload slots (\texttt{nonce1}, \texttt{nonce2}); principals are modelled by bounded interchangeable identifiers amenable to symmetry reduction. Honest principals are encoded as local role-state machines (e.g., \texttt{I\_SLEEP}, \texttt{I\_WAIT}, \texttt{I\_COMMIT}) with guarded transitions that enforce cryptographic checks before state advancement, providing the operational counterpart of the abstract deduction/equational layer in Section~\ref{sec:formalizing}.

\subsubsection{The Graded Intruder}
Consistent with Section~\ref{sec:threat_model}, the adversary is implemented as an asynchronous network process with full channel control. In the baseline Murphi style, intruder rules handle interception, storage, and message injection. Our extension adds \textbf{leakage-accumulation rules} that apply the Product T-norm updates to the discretised knowledge state \((\mu_K)\), tightening membership over target terms such as \(N_b\). Therefore, forgery rules that depend on \(N_b\) remain disabled until repeated leak updates push uncertainty below \(\alpha_{threshold}\).

\subsection{Implementation and Execution Boundary}
The main text keeps semantics and proofs, while execution-heavy implementation details are summarized here and expanded in Appendix~\ref{sec:exec_boundary_appendix}. The implementation uses a modified Murphi 3.1 with finite \texttt{real(m,n)} precision and external C/C++ fuzzy operations~\cite{dill1996murphi,intrigila2005fuzzy}; the external layer computes Gaussian fuzzy numbers, the Product T-norm intersection, alpha-cut support metrics, and entropy values over a finite discretization grid (101 levels on \([0,1]\), clamped) with key-domain mapping on \([0,255]\). Leak detection is based on alpha-cut support-width reduction (Hartley-style 1-bit event at 50\% support reduction), and executed precision sweeps (\texttt{real(4,2)} vs \texttt{real(4,4)}) preserve the same authentication-failure verdict in the fixed leaky NSL run. These settings define the bounded execution environment used in Section~\ref{sec:validation}.

Operationally, entropy and Hartley calculations are executed directly on the concrete state vector before the generation of a successor state. The external C++ layer computes the $\alpha$-cut support metrics continuously. To bridge the continuous entropy evaluation with the discrete transition system, leak detection is implemented via an $\alpha$-cut support-width reduction trigger. Specifically, a Hartley-style 1-bit event is registered in the state machine whenever the external layer detects a 50\% reduction in the support cardinality at $\alpha = 0.5$.

\subsection{The State-Space Reduction System}
To mitigate explicit-state growth, verification is performed on quotient spaces that preserve invariant outcomes.

First, we use Murphi's native symmetry reduction: agent/nonces declared as scalarset-like interchangeable identities induce state permutations that are graph automorphisms, so only one canonical representative per equivalence class is explored~\cite{dill1996murphi,ip1996symmetry,MMS97}.

Second, we apply exact concept-lattice reducts over the formal context $(U, V, R, \lambda)$. If $V' \subseteq V$ is an E-reduct, we map the explicit state space into the quotient space induced by the projection equivalence $\sim_{V'}$ established in Eq.~\eqref{eq:proj_equiv}. Under the encoded invariant class, pruning the redundant attributes $V \setminus V'$ and verifying strictly over this quotient space preserves encoded verdicts while significantly reducing state dimensionality (as formally guaranteed by Proposition~\ref{prop:reduct_preservation})).

\section{Validation and Case Studies}
\label{sec:validation}
To validate our graded verification framework, we evaluate the Needham--Schroeder--Lowe (NSL) public-key protocol~\cite{MMS97}. The protocol aims to establish mutual authentication and the symmetric exchange of fresh secrets between an initiator $A$ and a responder $B$ over an insecure network. The core execution logic is defined by the following three message exchanges:
\begin{align}
    A &\to B : \{N_a, A\}_{pk_B} \label{eq:nsl_step1} \\
    B &\to A : \{N_a, N_b, B\}_{pk_A} \label{eq:nsl_step2} \\
    A &\to B : \{N_b\}_{pk_B} \label{eq:nsl_step3}
\end{align}
Here, $N_a$ and $N_b$ represent freshly generated nonces, and $\{M\}_{pk}$ denotes asymmetric public-key encryption under the recipient's public key~\cite{MMS97}. Eq.~\eqref{eq:nsl_step2} incorporates Lowe's critical fix---the inclusion of the responder's identity $B$ inside the ciphertext---which formally guarantees that the classical Dolev--Yao intruder cannot successfully mount a man-in-the-middle impersonation attack.

To execute this abstraction within our modified Murphi backend, the term algebra $\mathcal{M}$ is mapped to finite, typed data structures, and the communication steps are encoded as asynchronous guarded commands~\cite{MMS97}. We strictly bound the symbolic universe by limiting the network capacity and using symmetry reduction (via interchangeability of agent identifiers~\cite{ip1996symmetry}) to mitigate state-space explosion prior to concept-lattice pruning. 

Within this explicit-state architecture, the classical DY intruder is encoded as an active network process. We extend this baseline by embedding our continuous epistemic tracking---specifically, the knowledge map $\mu_K$ and its associated Gaussian variance $\sigma_{leak}$---directly into the intruder's local state. This enables physical side-channel observations over the encrypted payloads (e.g., $\{N_b\}_{pk_B}$) to monotonically accumulate, ultimately triggering safe-to-fail boundaries without requiring the algebraic compromise of the underlying cryptographic primitives.

We now present E-reduct computation, automated verification of Needham--Schroeder--Lowe (NSL) for confidentiality and authentication, and notes on the framework's generalisation.

\subsection{Exact Reduct Computation}
To systematically evaluate the state-space optimization, we extracted the formal context $(U, V, R, \lambda)$ directly from the unreduced, full-state NSL specification. We then applied our concept-lattice reduction algorithm to this protocol-labeled dependency matrix. The encoded incidence matrix $R$ comprises $|U| = 10$ objects (representing the security-relevant transition rules and target invariants) mapped against $|V| = 17$ explicit state attributes. 

The computational reduction successfully identified $61$ distinct exact attribute reducts (E-reducts). By evaluating the intersection of these reducts, the algorithm isolates the core, indispensable state variables (e.g., \texttt{init\_responder} and \texttt{net.address}) while systematically flagging strictly redundant attributes. As guaranteed by Proposition 2, these redundant attributes can be safely pruned from the state vector prior to compilation, mapping the explicit state space into a lower-dimensional quotient space without altering the encoded verification verdicts.

Within the FRE/formal-context encoding, attributes represent protocol/state variables and objects represent equations/constraints. The reduct therefore captures the minimal variable subset preserving the discernibility of the encoded constraints~\cite{lobo2023fre,wille1982}. We then pruned the redundant attributes in the NSL model and measured the effect.

E-reducts mitigate dimensional state explosion while preserving verdicts for the encoded invariant class. In this NSL case, pruning it yields a median explored-state reduction of \(55.27\%\), consistent with Proposition~\ref{prop:reduct_preservation}. Table~\ref{tab:state_tradeoff} shows the normalised explored-state reduction.
Naming note: \texttt{nsl\_safe*} denotes the full-context/pruned reduct-benchmark variants, while \texttt{ns\_fuzzy\_*} denotes property-specific NSL verification variants used in the next sections.
\begin{table}[h]
\caption{Executed Pruning vs. Full-Context Tradeoff (3 Runs, Median Reported)}
\label{tab:state_tradeoff}
\begin{tabular}{l l l c c}
\toprule
\textbf{Model} & \textbf{States} & \textbf{Rules} & \textbf{Runtime (s)} & \textbf{State Range} \\
\midrule
\texttt{nsl\_safe\_fullctx} & 13067693 & 14446084 & 2.34 & 12.98M--13.07M \\
\texttt{nsl\_safe}          & 5845349  & 5885853  & 0.54 & 4.83M--5.93M \\
\bottomrule
\end{tabular}
\end{table}

\subsection{Evaluating Graded Authentication (NSL)}

Correctness conditions are encoded as Murphi invariants, consistent with the NS/NSL verification style used in early Murphi protocol research~\cite{MMS97}. We verified NSL safe and leaky variants using Murphi3.1. The crisp-safe baseline (\texttt{ns\_fuzzy\_auth\_safe}) explores $127,742$ states and reports no error, formally establishing that the classical transition system satisfies the crisp authentication invariant, $\mathcal{T}_{DY} \models \Phi_{auth}^{crisp}$. Under the strictly identical honest protocol logic, the leaky variant (\texttt{ns\_fuzzy\_auth}) yields a safe-to-fail transition ($\mathcal{T}_{FuzzyDY} \not\models \Phi_{auth}^{fuzzy}$) once accumulated side-channel observations enable the intruder's epistemic uncertainty to drop below the actionable threshold, allowing the deterministic forgery of the critical payload $\{Nb\}_{Kb}$.

Table~\ref{tab:attack_trace} provides the concrete witness for the safe-to-fail transition. It highlights coarse and fine leak updates, uncertainty reduction (\(\sigma_{leak}=4.96\) in the failing run), activation of the intruder-forge rule, and the final authentication violation.
\begin{table}[h]
\caption{Executed NSL Leaky Attack Trace (Condensed)}
\label{tab:attack_trace}
\begin{tabular}{l l c c c}
\toprule
\textbf{Step} & \textbf{Rule Fired} & \textbf{Leak Phase} & \textbf{$\sigma_{leak}$} & \textbf{Resp. State} \\
\midrule
1 & Leak (coarse) & coarse-done & 37.95 & sleep \\
2 & Leak (fine) & fine-done & 4.96 & sleep \\
3 & Initiator start (step 3) & fine-done & 4.96 & sleep \\
4 & Intruder intercept & fine-done & 4.96 & sleep \\
5 & Intruder replay to responder & fine-done & 4.96 & sleep \\
6 & Responder nonce reaction (3/6) & fine-done & 4.96 & wait \\
7 & Intruder intercept & fine-done & 4.96 & wait \\
8 & Intruder forge \{Nb\}Kb & fine-done & 4.96 & wait \\
9 & Responder nonce check (7) & fine-done & 4.96 & commit \\
\bottomrule
\end{tabular}
\end{table}

\subsection{Evaluating Graded Confidentiality (NSL)}
We execute confidentiality as a dedicated NSL pair (\texttt{ns\_fuzzy\_conf\_safe}, \texttt{ns\_fuzzy\_conf}) using the invariant \texttt{"Nb remains confidential from intruder"} encoded in Appendix~\ref{sec:murphi_invariants}. In the crisp baseline, the invariant holds. In the leaky variant, the verifier reports a direct confidentiality violation after staged coarse/fine observations and the subsequent responder message interception.
\begin{table}[h]
\caption{NSL Confidentiality Results from Executed Local Runs}
\label{tab:nsl_conf_results}
\begin{tabular}{l l l c}
\toprule
\textbf{Model} & \textbf{Variant} & \textbf{Result} & \textbf{States / Rules} \\
\midrule
\texttt{conf\_safe}  & Fixed NSL, no leak & No error found & 120359 / 122055 \\
\texttt{conf\_leaky} & Fixed NSL + leak rules & Invariant failed (\textit{Nb confidentiality}) & 133 / 132 \\
\bottomrule
\end{tabular}
\end{table}

In the violating state, \(\texttt{intr\_known\_nonce[B]}=\texttt{true}\) immediately after the intruder intercepts responder message 6 (\(\{N_a,N_b,B\}_{K_a}\)) with leak-evidence-enabled key knowledge, so confidentiality is lost before the later authentication-forgery step. Interpreted via Hartley non-specificity, this run gives a concrete leakage-budget crossing for \(N_b\): repeated sub-threshold observations drive uncertainty below the modelled secrecy boundary.

\subsection{Generalisation Across Protocol Topologies}
To test whether the observed threshold behaviour is specific to NSL, we evaluated a class of server-mediated symmetric-key authentication abstractions. We implemented models for Needham--Schroeder Symmetric Key (NSSK)~\cite{needham1978using} (see Table~\ref{tab:nssk_results}), and similarly for Yahalom~\cite{yahalom90}, Otway-Rees~\cite{otway1987efficient}, and Woo-Lam~\cite{woo-lam92}. In this template-equivalent family, crisp baseline variants verify as safe, while leaky variants trigger invariant violations after cumulative observations cross the operational threshold. Full execution results, complexity profiles, and failing traces for NSSK are provided in Appendix~\ref{sec:nssk_traces}.

\section{Discussion and Limitations}
\label{sec:discussion}
While the integration of concept-lattice E-reducts successfully mitigates dimensionality for the encoded invariants, our framework is subject to the universal bounds inherent to finite-state abstraction. We identify the following primary limitations:

\begin{itemize}
    \item \textbf{Scalability and State-Space Explosion:} The transition from a binary characteristic function to a graded epistemic knowledge map inherently exacerbates the state-space explosion problem. Discretizing fuzzy membership values introduces massive new state dimensions into the reachability graph. Although our application of concept-lattice E-reducts achieves a critical $55.27\%$ reduction in explored states for the NSL protocol by pruning redundant matrices, it mitigates rather than eliminates the exponential growth inherent to explicit-state model checking. Consequently, scaling this precise methodology to highly complex, multi-party industrial protocols remains a computational challenge.
    \item \textbf{Construct Validity of Leakage Parameterisation:} The assignment of initial fuzzy membership values, Gaussian variance $(\sigma_{leak})$, and the actionable vulnerability boundary $(\alpha_{threshold})$ currently relies on expert-driven parameterisation. While the Product T-norm accurately models the multiplicative contraction of uncertainty, transitioning this framework to real-world physical side-channel deployments may require dynamic, trace-derived parameterisation extracted directly from hardware execution logs.
    \item \textbf{Property Scope:} The formal guarantees presented in this paper are bounded to explicit-state symbolic assumptions (finite discretization and bounded sessions). Our current prototype semantics and validation restrict focus strictly to syntactic secrecy and trace-based authentication reachability invariants. Verifying broader indistinguishability properties (e.g., computational IND-CPA against probabilistic polynomial-time adversaries) or process-equivalence privacy claims falls outside the scope of this continuous-state symbolic capability.
\end{itemize}

\section{Conclusion and Future Work}
\label{sec:conclusion}
This paper introduced a graded extension of the Dolev--Yao model for side-channel-aware symbolic verification. The method augments binary abstraction with an epistemic adversary-view representation and is executed in a modified Murphi workflow. The central result is a leakage-threshold boundary: protocols that pass under crisp DY can fail once cumulative observations are enabled. In the NSL case, staged observations reduce candidate uncertainty and activate attack traces not present in the corresponding binary run.

We also showed that state growth from graded semantics can be mitigated through concept-lattice E-reducts, with a 55\% explored-state reduction on NSL while preserving encoded invariant verdicts. Across tested configurations, discretisation and precision sweeps produced stable verdicts for both authentication and confidentiality.

Future work will expand the verification scope beyond strict safety and correspondence assertions to encompass quantitative privacy analysis. This theoretical extension requires formalising a graded observational equivalence relation ($\approx_{fuzzy}$) as a continuous analogue to classical labelled bisimilarity ($\approx_l$), thereby enabling entropy-based unlinkability metrics over the adversary's epistemic view. Empirically, the framework will be scaled to evaluate the 5G Authentication and Key Agreement (5G-AKA) protocol~\cite{3gpp_5g_security}. Applying the E-reduct explicit-state workflow to the 5G-AKA abstraction will allow us to systematically characterise state-space scalability and measure protocol-specific leakage-tolerance boundaries under progressive, power-analysis-style physical observations.

\appendix

\section{Model Specification}
\label{sec:murphi_invariants}
\subsection{Execution Boundary Configuration}
\label{sec:exec_boundary_appendix}
\textbf{Compiler and type support.} The verifier uses a local modified Murphi 3.1 toolchain, based on the original Murphi line~\cite{dill1996murphi}. The compiler supports finite-precision \texttt{real(m,n)} handling and external C/C++ linkage for fuzzy operations, following the extension pattern used in prior fuzzy-control model-checking adaptations~\cite{intrigila2005fuzzy}. Executed models use \texttt{real(4,2)} and \texttt{real(4,4)} variants.

\textbf{External fuzzy-logic boundary.} Complex fuzzy operations are implemented in C++ and invoked from Murphi through external declarations. The external layer implements Gaussian fuzzy-number representation, the Product T-norm intersection, discrete entropy and alpha-cut support computations, and approximate fuzzy-encrypt uncertainty propagation~\cite{klir1995fuzzy,zadeh1965fuzzy}.

\textbf{Discretization and leak metric.} Membership degrees are discretised to 101 levels on $[0,1]$ with clamping; key-domain values are mapped to $[0,255]$. This is a cell-to-cell mapping abstraction where each continuous update is projected to a finite cell before successor generation. Leak tracking uses entropy trend and support shrinkage at \(\alpha=0.5\), with Hartley interpretation \(H=\log_2|A|\) ~\cite{hartley1928,dubois1980fuzzy}.

\textbf{Precision robustness.} The Product T-norm updates use closed-form Gaussian-product parameter fusion followed by parameter quantisation (instead of repeated pointwise multiplication on discretised samples). Executed precision sweeps (\texttt{real(4,2)} vs \texttt{real(4,4)}) preserve the same fixed-leaky NSL authentication-failure verdict and explored-state counts.

\subsection{Exact Murphi Invariants Used in Executed Models}
\subsubsection{Public-Key Needham--Schroeder Authentication Model}
\begin{verbatim}
invariant "responder correctly authenticated" 
  forall i: InitiatorId do
    init_state = I_COMMIT &
    init_responder = B
    ->
    resp_initiator = i &
    (resp_state = R_WAIT | resp_state = R_COMMIT)
  end;

invariant "initiator correctly authenticated"
  forall r: ResponderId do
    resp_state = R_COMMIT &
    resp_initiator = A
    ->
    init_state = I_COMMIT &
    init_responder = r
  end;

invariant "one-bit flag implies at least one leak step"
leakOneBit -> leakPhase != LP_NONE;
\end{verbatim}

\subsubsection{Public-Key Needham--Schroeder Confidentiality Model}
\begin{verbatim}
invariant "responder correctly authenticated" 
  forall i: InitiatorId do
    init_state = I_COMMIT &
    init_responder = B
    ->
    resp_initiator = i &
    (resp_state = R_WAIT | resp_state = R_COMMIT)
  end;

invariant "initiator correctly authenticated"
  forall r: ResponderId do
    resp_state = R_COMMIT &
    resp_initiator = A
    ->
    init_state = I_COMMIT &
    init_responder = r
  end;

invariant "Nb remains confidential from intruder"
!intr_known_nonce[B];

invariant "one-bit flag implies at least one leak step"
leakOneBit -> leakPhase != LP_NONE;
\end{verbatim}

\subsubsection{Symmetric-Key Families}
\begin{verbatim}
invariant "responder authenticated initiator"
  forall i: InitiatorId do
    b_state = RB_COMMIT ->
      a_state = IA_COMMIT &
      a_partner = B &
      b_initiator = i
  end;

invariant "initiator commits only for responder B"
  forall r: ResponderId do
    a_state = IA_COMMIT -> a_partner = r
  end;

invariant "one-bit flag implies at least one leak"
  leakOneBit -> leakPhase != LP_NONE;
\end{verbatim}

\section{Extended Results: Symmetric-Key Families}
\label{sec:symmetric_val}
\subsection{Executed Symmetric-Key Results}
\begin{table}[h]
\caption{Executed NSSK Results}
\label{tab:nssk_results}
\begin{tabular}{l l l c }
\toprule
\textbf{Model} & \textbf{Variant} & \textbf{Result} & \textbf{States / Rules} \\
\midrule
\texttt{nssk\_safe} & Baseline (no leak rules) & No error found & 36 / 35 \\
\texttt{nssk\_leaky} & Side-channel leak + accumulation & Invariant failed (auth.) & 1404 / 1403 \\
\bottomrule
\end{tabular}
\end{table}

\subsection{Symmetric-Key Failing Traces}
\label{sec:nssk_traces}
\subsubsection{NSSK Leaky Model Full Trace}
\begin{verbatim}
1   intruder side-channel leak (coarse session-key reading)
2   intruder side-channel leak (fine session-key reading)
3   A starts NSSK request to server
4   intruder intercepts
5   intruder replays recorded message to server
6   Server responds with package to initiator
7   intruder intercepts
8   intruder replays recorded message to initiator
9   Initiator processes server package and sends ticket
10  intruder intercepts
11  intruder replays recorded message to responder
12  Responder processes ticket and sends challenge
13  intruder intercepts
14  intruder forges challenge response after leakage
15  Responder verifies challenge response
\end{verbatim}

\subsubsection{Additional Family Leaky Traces (Pattern Equivalence)}
The executed leaky traces for Yahalom, Otway-Rees, and Woo-Lam are template-equivalent to the NSSK trace above under the current abstraction: each follows the same 15-step control pattern (coarse leak, fine leak, request/relay exchanges, then leakage-enabled forged challenge response and responder commit). Differences are limited to family labels and leak-parameter constants, so repeated full rule listings are omitted.

\section{Supplemental Derivations}
\subsection{Worked Alpha-Cut Interval Propagation Example}
Consider a scenario where the $\alpha=0.5$ cut of a targeted nonce interval evaluates to $[38.23, 61.77]$. For a monotone constructor $f(x) = x + 10$, interval propagation gives $[48.23, 71.77]$. Following conservative integer-grid discretization via $\lfloor \cdot \rfloor$ and $\lceil \cdot \rceil$, this maps to the finite cell $[4, 5]$, yielding a discrete support cardinality of $25$. The resulting Hartley non-specificity is calculated as $H = \log_2(25) \approx 4.64$ bits. As the execution progresses and the intruder accumulates further leak observations, the system dynamically monitors this discrete entropy. A discrete "leak event" rule fires exclusively when repeated sub-threshold observations drive the uncertainty down, causing $|\mathrm{Supp}^{0.5}|$ to drop by 50\% (representing a 1-bit reduction in $H$).

\bibliographystyle{plain}
\bibliography{references}

\end{document}